\documentclass[lettersize,journal]{IEEEtran}
\usepackage{amsmath,amsfonts}
\usepackage{algorithmic}
\usepackage{algorithm}
\usepackage{array}
\usepackage[caption=false,font=normalsize,labelfont=sf,textfont=sf]{subfig}
\usepackage{textcomp}
\usepackage{stfloats}
\usepackage{url}
\usepackage{verbatim}
\usepackage{graphicx}
\usepackage{cite}
\usepackage{amsthm}

\newtheorem{theorem}{Theorem}[section]
\newtheorem{lemma}{Lemma}[section]
\newtheorem{corollary}[theorem]{Corollary}

\newtheorem{definition}{Definition}[section]

\begin{document}

\title{Distributed Matrix Multiplication with a Smaller Recovery Threshold through Modulo-based Approaches} 

\author{Zhiquan Tan$^*$, Dingli Yuan$^*$, Zihao Wang$^\dagger$, Zhongyi Huang$^*$\\
$*$ Department of Mathematical Sciences, Tsinghua University \\
$\dagger$ Department of CSE,
HKUST \\
\thanks{This work was partially supported by the NSFC Project No. 12025104.}

}

\maketitle

\begin{abstract}
This paper considers the problem of calculating the matrix multiplication of two massive matrices $\mathbf{A}$ and $\mathbf{B}$ distributedly. We provide a modulo technique that can be applied to coded distributed matrix multiplication problems to reduce the recovery threshold. This technique exploits the special structure of interpolation points and can be applied to many existing coded matrix designs.  Recently studied discrete Fourier transform based code achieves a smaller recovery threshold than the optimal MatDot code with the expense that it cannot resist stragglers. We also propose a distributed matrix multiplication scheme based on the idea of locally repairable code to reduce the recovery threshold of MatDot code and provide resilience to stragglers. We also apply our constructions to a type of matrix computing problems, where generalized linear models act as a special case.
\end{abstract}



%
\IEEEpeerreviewmaketitle



\section{Introduction}
\IEEEPARstart{I}{n} the big data era, the computation of large-scale matrix multiplication is inevitable due to the prevalence of matrix algebra in many important algorithms. As the matrix gets bigger, it is difficult to carry out the multiplication on a single server \cite{buluc2008challenges}. Thus it is common to use distributed servers jointly to compute the result. But when spreading the computation over many workers, there will be a straggler issue \cite{dean2013tail}, namely some worker nodes may return their computation result significantly slower than others or even fail to respond. There are also cases when some of the workers are curious, they may collude and seek to infer information about the encoded data. Security in this case means that there shall be no information leak even when any $X$ workers collude. Coded computing \cite{li2020coded} is a type of method that uses coding theoretic techniques to inject redundancy in the computation to tackle the straggler effect and maintain security.

More precisely, there will be a master who is interested in multiplying two massive matrices $\mathbf{A}$ and $\mathbf{B}$ distributedly with the aid of several workers. Lee et al. \cite{lee2017speeding} are the first to introduce the coding theoretic techniques in distributed computing to mitigate stragglers.  They use a maximal distance separable (MDS) code to inject redundancy into the computation, thus mitigating the straggler effect. Since then, using codes to help distributed matrix multiplication has been very popular, especially using polynomial-based codes. Typical design involves designing two polynomial functions $f(x)$ and $g(x)$ that are generated by matrices $\mathbf{A}, \mathbf{B}$, or even random matrices. The master then chooses some distinct numbers $x_i$. For each worker $i$, it will be assigned two encoded matrices $f(x_i)$ and $g(x_i)$. Each worker will then multiply $f(x_i)g(x_i)$, and the desired matrix product $\mathbf{A}\mathbf{B}$ can be derived from one or some of the product polynomial $f(x)g(x)$'s coefficients. The recovery threshold is defined as the minimum number of responses that the master receives from the workers to deduce the matrix product $\mathbf{A}\mathbf{B}$. It is closely related to the degree of product polynomial $f(x)g(x)$. When $\mathbf{A}$ is partitioned row-wise and $\mathbf{B}$ partitioned column-wise, Yu et al. \cite{yu2017polynomial} propose Polynomial code that constructs $f(x)$ and $g(x)$ with the partitioning blocks as coefficients. Then Dutta et al. propose MatDot code \cite{dutta2019optimal} for inner product partitioning, i.e. column-wise partitioning of matrix $\mathbf{A}$ and row-wise partitioning of $\mathbf{B}$. MatDot code reduces the recovery threshold compared to Polynomial code, with the expense of higher computational cost at each server. If both $\mathbf{A}$ and $\mathbf{B}$ admit arbitrarily (block) partitioning, entangled polynomial code \cite{yu2020straggler} and Generalized PolyDot \cite{dutta2018unified} give more flexible trade-off between recovery threshold and computation overhead.

Security is another major concern in designing distributed matrix multiplication schemes. \cite{chang2018capacity} first propose a one-sided secure scheme, namely keeping one of the matrices secure. Designing a scheme that is two-sided secure has been very popular. \cite{d2020gasp} proposes Gap Additive Secure Polynomial (GASP) code for polynomial code-like partitioning of matrices and achieves two-sided security. \cite{SDMMDFT} proposes a discrete Fourier transform based construction for inner product partitioning of matrices. There also have been works that deal with block partitioned coded matrix multiplication \cite{zhu2021improved,nodehi2018entangled,yu2020entangled,aliasgari2020private, li2022private}. 

Locally repairable codes (LRCs) are a kind of erause codes where a codeword symbol can be determined from a small set of coordinates, which are more efficient than classic MDS codes in terms of the communication cost during the regeneration of a codeword symbol. The $i$ th code symbol $c_{i}$ is said to have $(r, \delta)$ locality \cite{prakash2012optimal} if there exists an index set $S_{i}$ containing $i$ such that $\left|S_{i}\right| \leq r+\delta-1$ and each symbol $c_{j} ( j \in S_{i})$, can be reconstructed by any other $r$ symbols in $S_{i}.$ LRCs have been deployed in distributed storage \cite{prakash2012optimal,huang2012erasure} due to their small communication cost when repairing a failed node. Optimal constructions of LRCs can be found in \cite{tamo2014family} and the references therein. Recently, LRC has been introduced in coded matrix design
to help perform consecutive matrix multiplications \cite{jeong2018locally}.

There has also been works that exploit the structure of interpolation points \cite{fahim2021numerically, rudow2020locality, SDMMDFT}. To better demonstrate the effect of interpolation points, we generalize the notion of recovery threshold to include the case that some special subsets with smaller number of nodes may also be able to retrieve the desired matrix computation result. In this paper, we introduce a modulo-based approach that leads to better candidates for interpolation point selection. We apply it to many existing coded matrix designs and reduce their recovery threshold. Moreover, we propose a locality-based coded matrix design that can be combined with a modulo trick to give a reduced recovery threshold. We also discuss the applications of our proposed method to a type of matrix computing problem with a special structure, which includes problems like generalized linear models as special cases.  

The paper is organized as follows. Section \ref{sec:setting} introduces the problem settings. Section \ref{sec:prelim} discusses the needed knowledge that will be used in the rest of the paper, especially the modulo technique. Section \ref{sec:apply} explains how the modulo techniques benefit coded matrix design. Section \ref{structure mat mult} applies our constructions to a type of structured matrix multiplication problem. Section \ref{GLM} uses the generalized linear models as a special example of the problem introduced in section \ref{structure mat mult}. Section \ref{sec:compar} compares our method with existing state-of-the-art methods. Section \ref{sec:conclu} concludes this paper.

\section{Problem Setting} \label{sec:setting}
We shall follow a master-worker setting, there will be one powerful master node and $N$ worker nodes with limited storage capacity. Our goal is to parallel the costly process of computation matrix multiplication of matrix $\mathbf{A}$ and $\mathbf{B}$ to the workers. Matrix $\mathbf{A}$ will have shape $a \times b$, matrix $\mathbf{B}$ will have shape $b \times c$. We assume all matrices are on a finite field, the complex number field $\mathbf{C}$ case is similar.

\section{Preliminary} \label{sec:prelim}
In this section, we shall go through the basic concepts that we shall use in the paper.

\subsection{Secure Entangled Polynomial Code}
\label{block partition}

We shall partition the two matrices as follows,
\begin{equation}
 \mathbf{A}=\left[\begin{array}{ccc}
\mathbf{A}_{1,1} & \cdots & \mathbf{A}_{1, m} \\
\vdots & \ddots & \vdots \\
\mathbf{A}_{K_1,1} & \cdots & \mathbf{A}_{K_1, m}
\end{array}\right], 
\end{equation}
\begin{equation}
 \mathbf{B}=\left[\begin{array}{ccc}
\mathbf{B}_{1,1} & \cdots & \mathbf{B}_{1, K_2} \\
\vdots & \ddots & \vdots \\
\mathbf{B}_{m,1} & \cdots & \mathbf{B}_{m, K_2}
\end{array}\right].
\end{equation}

As the goal is to compute $\mathbf{A}\mathbf{B}$ distributedly and securely. Each worker $i$ will be sent encoded matrices $\tilde{\mathbf{A}}_i$ and $\tilde{\mathbf{B}}_i$. We shall then introduce the definition of security.

\begin{definition}
A distributed matrix multiplication scheme for calculating $\mathbf{A}\mathbf{B}$ is (information theoretic) $X$-secure if
\begin{equation}
\mathbf{I}\left(\mathbf{A}, \mathbf{B} ; \widetilde{\mathbf{A}}_{\mathcal{X}}, \widetilde{\mathbf{B}}_{\mathcal{X}}\right) = 0, 
\end{equation}
for any $\mathcal{X} \subset[N],|\mathcal{X}|=X$. $\widetilde{\mathbf{A}}_{\mathcal{X}} = \left\{\tilde{\mathbf{A}}_i\right\}_{i \in \mathcal{X}}$ denotes the encoded matrices indexed by set $\mathcal{X}$, $\widetilde{\mathbf{B}}_{\mathcal{X}}$ defines similarly.
\end{definition}

The matrix product can be written as follows,
\begin{equation}
\mathbf{C}=\left[\begin{array}{cccc}
\mathbf{C}_{1,1} & \mathbf{C}_{1,2} & \cdots & \mathbf{C}_{1, K_2} \\
\mathbf{C}_{2,1} & \mathbf{C}_{2,1} & \cdots & \mathbf{C}_{2, K_2} \\
\vdots & \vdots & \ddots & \vdots \\
\mathbf{C}_{K_1,1} & \mathbf{C}_{K_1,2} & \cdots & \mathbf{C}_{K_1, K_2}
\end{array}\right],    
\end{equation}
where each submatrix $\mathbf{C}_{i, j} = \sum_{k=1}^{m} \mathbf{A}_{i, k} \mathbf{B}_{k, j}$.

Assume $\mathbf{R}_i$ and $\mathbf{T}_i$ ($1 \leq i \leq X$) are random matrices that are sampled uniformly. Then polynomial-based encoding functions are given by

\begin{align}
p_{\mathbf{A}}(x) & =\sum_{j=1}^{K_1} \sum_{k=1}^{m} \mathbf{A}_{j, k} x^{\alpha_{j,k}} + \sum^{X}_{i=1} \mathbf{R}_i x^{\theta_i}, \\
p_{\mathbf{B}}(x) & =\sum_{j=1}^{m} \sum_{k=1}^{K_2} \mathbf{B}_{j, k} x^{\beta_{j,k}} + \sum^{X}_{i=1} \mathbf{T}_i x^{\eta_i}.
\end{align}

The product polynomial is given as:

\begin{align*}
  p_{\mathbf{A}}(x)p_{\mathbf{B}}(x) &=  \sum_{j=1}^{K_1} \sum_{k=1}^{m}\sum_{j'=1}^{m} \sum_{k'=1}^{K_2} \mathbf{A}_{j, k}\mathbf{B}_{j', k'} x^{\alpha_{j,k}+\beta_{j',k'}}    \\
  &+ \sum^{X}_{i=1}\sum_{j=1}^{m} \sum_{k=1}^{K_2}  \mathbf{R}_i \mathbf{B}_{j, k} x^{\theta_i+\beta_{j,k}} \\
  &+ \sum_{j=1}^{K_1} \sum_{k=1}^{m} \sum^{X}_{i=1} \mathbf{A}_{j, k}\mathbf{T}_i x^{\alpha_{j,k}+\eta_i} \\
  &+ \sum^{X}_{i=1}\sum^{X}_{i'=1} \mathbf{R}_i\mathbf{T}_{i'} x^{\theta_i+\eta_{i'}}.
\end{align*}

The goal is to choose suitable parameters $\alpha,\beta,\theta,\eta$ such that each $C_{i, j}$ is in the coefficient of product polynomial $p_{\mathbf{A}}(x)p_{\mathbf{B}}(x)$. Many state-of-the-art works follow this encoding polynomial designing strategy.

When $X>0$.
\begin{enumerate}
    \item Choose $\alpha_{j,k}=(k-1)+(j-1)m$, $\beta_{j,k}=(m-j)+(k-1)(K_1m+X)$, $\theta_i=K_1m+i-1$ and $\eta_i=(K_2-1)(K_1m+X)+K_1m+i-1$ or \\
    $\alpha_{k,j}=(j-1)+(k-1)(K_2m+X)$, $\beta_{j,k}=(m-j)+(j-1)m$, $\theta_i=(K_1-1)(K_2m+X)+K_2m+i-1$ and $\eta_i=K_2m+i-1$ will recover the secure entangled polynomial (SEP) code \cite{zhu2021improved}. The recovery threshold is given by $min \{(K_2+1)(K_1m+X)- 1, (K_1+1)(K_2m+X)- 1 \}$. When $K_1=K_2=1$ the secure entangled polynomial code is exactly the secure MatDot code\cite{aliasgari2020private}.
   \item Choose $\alpha_{j,k}=(k-1)+(j-1)m$, $\beta_{j,k}=(m-j)+(k-1)K_1m$, $\theta_i=K_1K_2m+i-1$ and $\eta_i=K_1K_2m+i-1$ will recover the polynomial sharing (PS) code \cite{nodehi2018entangled}. The recovery threshold is given by $2K_1K_2m+2X-1$. 
\end{enumerate}

When $X=0$.
\begin{enumerate}
    \item Choose $\alpha_{k,j}=(k-1)+(j-1)m$ and $\beta_{j,k}=(m-j)+(k-1)K_1m$ will recover the entangled polynomial code \cite{yu2020straggler}. The recovery threshold of entangled polynomial (EP) code is given by $K_1K_2m + m-1$. 
    \item  When $K_1=K_2=1$, the entangled polynomial code is MatDot code\cite{dutta2019optimal}. When $m=1$, the entangled polynomial code is Polynomial code\cite{yu2017polynomial}.
\end{enumerate}

Suppose there are $N$ distinct numbers $x_i$. The master shall assign each worker $i$ $\tilde{\mathbf{A}}_i=p_{\mathbf{A}}(x_i)$ and $\tilde{\mathbf{B}}_i=p_{\mathbf{B}}(x_i)$. Each worker $i$ shall compute $\tilde{\mathbf{A}}_i\tilde{\mathbf{B}}_i$ and return the result to the master. Once the master receives $deg(p_{\mathbf{A}}(x)p_{\mathbf{B}}(x))+1$ responses from workers, it can interpolate this polynomial and get all the coefficients to recover the matrix product $\mathbf{A}\mathbf{B}$.

We shall formally give the definition of recovery threshold as follows:

\begin{definition}
An integer $K$ is called the best recovery threshold for a distributed matrix multiplication scheme for calculating $\mathbf{A}\mathbf{B}$ if and only if $K$ is the minimum number that satisfies the following requirements on condition entropy: 
\begin{equation}
\mathbf{H}\left(\mathbf{A} \mathbf{B} \mid \{\widetilde{\mathbf{A}}_{i} \widetilde{\mathbf{B}}_{i}\}_{i \in \tau}\right) = 0,  
\end{equation}
for \textbf{some} subset of nodes $\mathcal{\tau}$. 

An integer $\bar K$ is called the worst recovery threshold for a distributed matrix multiplication scheme for calculating $\mathbf{A}\mathbf{B}$ if and only if $\bar K$ is the minimum number that satisfies the following requirements on condition entropy: 
\begin{equation}
\mathbf{H}\left(\mathbf{A} \mathbf{B} \mid \{\widetilde{\mathbf{A}}_{i} \widetilde{\mathbf{B}}_{i}\}_{i \in \bar \tau}\right) = 0,   
\end{equation}
for \textbf{any} subset of nodes $\mathcal{\bar \tau}$. 
\end{definition}

\subsection{Discrete Fourier Transform Matrix }

For any polynomial $p(x)=\sum^{K-1}_{i=0} a_{K-1-i}x^i$ of degree $K-1$, we shall first define the coefficient vector $a=(a_{K-1},\cdots,a_{0})$. Then we shall denote the evaluation value vector at $K$ distinct points $x_i$ ($1 \leq i \leq K$) as $f=(f(x_1),\cdots,f(x_K))$. It is straightforward to see that $f =a V(x_1,\cdots,x_K)$, where the matrix $V(x_1,\cdots,x_K)$ is the vandermonde matrix. 

Denote the primitive $k$-th root of unity as $\zeta$. We take the evaluation points $x_i = \zeta^i$, the Vandermonde matrix is exactly the discrete Fourier transform Matrix \cite{sundararajan2001discrete}. The discrete Fourier transform matrix $V$ has a special structure, namely its inverse matrix can be calculated in closed form $V^{-1}= \frac{1}{K}V(x^{-1}_1,\cdots,x^{-1}_K)$ \cite{preparata1977computational}. When the evaluation points are $(r\zeta)^i$, where $r$ is an arbitrary non-zero element, the inverse matrix has a similar closed form. 

We shall next present two lemmas that are closely related to the discrete Fourier transform matrix, which we shall call these lemmas modulo techniques.

\begin{lemma}
\label{lemma 1}
For any polynomial $A(x)= \sum^{n}_{i=0} a_{n-i} x^{i}$ of degree $n$, integer $k$ such that $n \geq k \geq n-s+1$, where $s \leq \frac{n}{2}$ is a given integer. When interpolating on the $k$-th root of unity set $\{ x^{k} = 1  \}$, we can effectively get the coefficients of $x^{j}$ ($s \leq j \leq n-s$) by the closed form expression of inverse discrete fourier transform matrix.
\end{lemma}

\begin{proof}
From Euclidean division theorem, we shall know that there exists quotient polynomial $Q(x) = \sum^{n-k}_{i=0} a_{n-k-i} x^{i}$ and remainder polynomial $R(x) =  \sum^{k-1}_{i=n+1-k} a_{n-i} x^{i} + \sum^{n-k}_{i=0} ( a_{n-k-i} +  a_{n-i}) x^{i}$ which satisfy the relation $A(x)=(x^{k}-1)Q(x)+R(x)$. $R(x)$ is exactly the modulo result when $A(x)$ modulo $x^k - 1$. Note the coefficient of $x^{j}$ ($n+1-k \leq j \leq k-1$) in $R(x)$ remains the same as $A(x)$. As $n \geq k \geq n-s+1$, we know $k-1 \geq n-s$ and $n+1-k \leq s$. Thus the coefficient of $x^{j}$ ($s \leq j \leq n-s$) in $R(x)$ remains the same as $A(x)$.

For any $x_i \in \{ x^{k} = 1  \}$, we know that $A(x_i)=R(x_i)$. Whenever we get all the $k$ evaluation results of $A(x)$ at the set $\{ x^{k} = 1  \}$, it is equivalent to getting the evaluation results of $R(x)$. Since the degree of $R(x)$ is $k-1$, we can interpolate the whole polynomial $R(x)$, thus getting the coefficient of $x^{j}$ ($s \leq j \leq n-s$), which is the same for $A(x)$. Note the interpolation points are just the $k$-th root of unity set, each coefficient can be decoded effectively by the closed form expression of inverse discrete fourier transform matrix.

\end{proof}

\textbf{Remark 1}: It is worth noticing the case when we get the evaluation results on $\{ x^{k} = \gamma  \}$, where $\gamma$ is an arbitrary nonzero element. Note $\{ x^{k} = \gamma  \}$ is just a multiplicative coset of the $k$-th root of unity set, we have similar results as lemma \ref{lemma 1}.

Sometimes we are interested in multiple coefficients with a special pattern, we shall next derive the second lemma that aims at getting them from the interpolation results at the root of unity set.

\begin{lemma}
\label{lemma 2}
For any uni-variate polynomial $A(x)= \sum^{(k+1)p-2}_{i=0} a_{(k+1)p-2-i} x^{i}$ of degree $(k+1)p-2$ ($k \geq 1$), redefine $b_j=a_{(k+1)p-2-jp-p+1}$ ($j= 0, \cdots,k-1$). Define another uni-variate polynomial as $B(x) = \sum^{k-1}_{j=0} b_j x^{j}$. Then for any nonzero element $\gamma$, when interpolating on the set $\{ x^{p} = \gamma  \}$, we can effectively get the evaluation of $B(x)$ at $\gamma$. A direct implication is that we can get coefficients $b_j$s from the evaluation results of $A(x)$ on the $kp$-th root of unity set.   
\end{lemma}

\begin{proof}
Note when interpolating on the set $\{x^{p} = \gamma  \}$, from a similar argument of lemma \ref{lemma 1}, we can get the coefficient of $x^{p-1}$ in the remainder polynomial i.e. $\sum^{k-1}_{j=0} (\gamma)^{j}b_j$. From the definition of $B(x)$, this is exactly $B(\gamma)$. 

If we denote $\zeta$ as the primitive $k$-th root of unity, then taking $\gamma_i = \zeta^i$ will make $\{ x^{p} = \gamma_i  \}$ ($i= 0, \cdots,k-1$) exactly recover the $kp$-th root of unity set. From the evaluation results of $A(x)$ on the $kp$-th root of unity, we can get $B(\gamma_i)$s easily. As what we want are just the coefficients of $B(x)$, we can decode them efficiently from fast Fourier transform \cite{preparata1977computational}.
\end{proof}

\section{Modulo techniques for coded matrix matrix multiplication} \label{sec:apply}

\subsection{Apply Modulo Techniques to Secure Matrix Multiplication} \label{SEP}

\subsubsection{Apply modulo techniques to existing secure matrix multiplication schemes}

The modulo techniques introduced in lemma \ref{lemma 1} can be applied to secure entangled polynomial code and polynomial sharing code \ref{block partition}.

\begin{theorem} \label{SEP theorem}
    Denote $\hat A=min \{(K_2+1)(K_1m+X)- 1, (K_1+1)(K_2m+X)- 1, 2K_1K_2m+2X-1 \}$. When $N \geq \hat A$, assume the $\hat A-m+1$-th root of unity set is a subset of the evaluation set. Then the worst recovery threshold is $\hat A$, and the best recovery threshold is $\hat A-m+1$.
    
\end{theorem}

\begin{proof}
The encoding functions are the same as secure entangled polynomial code and polynomial sharing code, thus having the worst recovery threshold of $\hat{A}$. 

We shall next present the product polynomial $p_{\mathbf{A}}(x)p_{\mathbf{B}}(x)$ of one of the constructions in \cite{zhu2021improved}, the other cases are similar.

\begin{align*}
  &p_{\mathbf{A}}(x)p_{\mathbf{B}}(x) \\
  =&  \sum_{j=1}^{K_1} \sum_{k=1}^{m}\sum_{j'=1}^{m} \sum_{k'=1}^{K_2} \\
  &\mathbf{A}_{j, k}\mathbf{B}_{j', k'} x^{(k-1)+(j-1)m+(m-j')+(k'-1)(K_1m+X)}    \\
  &+ \sum^{X}_{i=1}\sum_{j=1}^{m} \sum_{k=1}^{K_2}  \mathbf{R}_i \mathbf{B}_{j, k} x^{K_1m+i-1+(m-j)+(k-1)(K_1m+X)} \\
  &+ \sum_{j=1}^{K_1} \sum_{k=1}^{m} \sum^{X}_{i=1} \\ 
  &\mathbf{A}_{j, k}\mathbf{T}_i x^{(k-1)+(j-1)m+(K_2-1)(K_1m+X)+K_1m+i-1} \\
  &+ \sum^{X}_{i=1}\sum^{X}_{i'=1} \mathbf{R}_i\mathbf{T}_{i'} x^{K_1m+i-1+(K_2-1)(K_1m+X)+K_1m+i'-1}.
\end{align*}

Note $\mathbf{C}_{i, j} = \sum_{k=1}^{m} \mathbf{A}_{i, k} \mathbf{B}_{k, j}$, it is the coefficient of the $m-1+(i-1)m+(j-1)(K_1m+X)$-th degree term. The first $m-1$ coefficients in $p_{\mathbf{A}}(x)p_{\mathbf{B}}(x)$ are not needed and the last $m-1$ coefficients are random noises, thus can be alleviated by modulo. Applying lemma \ref{lemma 1} to each element of the product matrix polynomial will give the best recovery threshold of $\hat{A}-m+1$. As the degrees of random matrices are consecutive, the security follows from \cite{chang2018capacity}.
\end{proof}

\textbf{Remark 2}: When $K_1=1, K_2=1$, the best recovery threshold is $m+2X$, matching the one given by discrete Fourier transform based code \cite{SDMMDFT}.

\textbf{Example 1}: ($K_1=K_2=m=2, X=2, N=16$)

\begin{equation*}
 \mathbf{A}=\left[\begin{array}{cc}
\mathbf{A}_{1,1} &  \mathbf{A}_{1, 2} \\
\mathbf{A}_{2,1} & \mathbf{A}_{2, 2}
\end{array}\right], 
\end{equation*}
\begin{equation*}
 \mathbf{B}=\left[\begin{array}{cc}
\mathbf{B}_{1,1} &  \mathbf{B}_{1, 2} \\
\mathbf{B}_{2,1} & \mathbf{B}_{2, 2}.
\end{array}\right]
\end{equation*}

It is clear that 
\begin{equation*}
  \mathbf{A}\mathbf{B} =   \left[\begin{array}{cc}
\mathbf{A}_{1,1}\mathbf{B}_{1,1}+\mathbf{A}_{1,2}\mathbf{B}_{2,1} &  \mathbf{A}_{1,2}\mathbf{B}_{2,2}+\mathbf{A}_{1,1}\mathbf{B}_{1,2} \\
\mathbf{A}_{2,1}\mathbf{B}_{1,1}+\mathbf{A}_{2,2}\mathbf{B}_{2,1} & \mathbf{A}_{2,1}\mathbf{B}_{1,2}+\mathbf{A}_{2,2}\mathbf{B}_{2,2}
\end{array}\right].
\end{equation*}

From the 

The encoding polynomials are as follows:
\begin{align*}
&p_{\mathbf{A}}(x) \hspace{-1mm} = \hspace{-1mm}\mathbf{A}_{1,1}\hspace{-1mm}  +\hspace{-1mm}\mathbf{A}_{1,2}x\hspace{-1mm}  +\hspace{-1mm}\mathbf{A}_{2,1}x^{2}\hspace{-1mm}  +\hspace{-1mm}\mathbf{A}_{2,2}x^{3}\hspace{-1mm}  +\hspace{-1mm}\mathbf{R}_1 x^4\hspace{-1mm}  +\hspace{-1mm}\mathbf{R}_2 x^5. \\
& p_{\mathbf{B}}(x) \hspace{-1mm} = \hspace{-1mm}\mathbf{B}_{1,1}x\hspace{-1mm}  +\hspace{-1mm}\mathbf{B}_{2,1}\hspace{-1mm}  +\hspace{-1mm}\mathbf{B}_{1,2}x^{7}\hspace{-1mm}  +\hspace{-1mm}\mathbf{B}_{2,2}x^{6}\hspace{-1mm}  +\hspace{-1mm}\mathbf{T}_1 x^{10}\hspace{-1mm}  +\hspace{-1mm}\mathbf{T}_2 x^{11}.
\end{align*}

Thus the product $p_{\mathbf{A}}(x)p_{\mathbf{B}}(x)$ shall rewritten as follows,
\begin{align*}
& \mathbf{A}_{1,1}\mathbf{B}_{2,1} +(\mathbf{A}_{1,1}\mathbf{B}_{1,1}+\mathbf{A}_{1,2}\mathbf{B}_{2,1})x +   \\& (\mathbf{A}_{1,2}\mathbf{B}_{1,1}+\mathbf{A}_{2,1}\mathbf{B}_{2,1})x^2 
+(\mathbf{A}_{2,1}\mathbf{B}_{1,1}+\mathbf{A}_{2,2}\mathbf{B}_{2,1})x^3 + \cdots  \\
 &
 +(\mathbf{A}_{1,1}\mathbf{B}_{1,2} + \mathbf{A}_{1,2}\mathbf{B}_{2,2})x^7 + (\mathbf{A}_{2,1}\mathbf{B}_{2,2} + \mathbf{A}_{1,2}\mathbf{B}_{1,2})x^8 \\
&+(\mathbf{A}_{2,1}\mathbf{B}_{1,2}+\mathbf{A}_{2,2}\mathbf{B}_{2,2})x^9 + \cdots + \mathbf{R}_2 \mathbf{T}_2 x^{16}.
\end{align*}

The polynomial is of degree $16$, the first and last coefficients are noises. The needed degree is $1,3,7,9$. Thus interpolating at the $16$-th root of unity set is sufficient to get the desired matrix product $\mathbf{A}\mathbf{B}$ by lemma \ref{lemma 1}.

\subsubsection{A new construction of secure matrix multiplication} \label{transpose sep code}

We discussed how to apply our modulo technique to existing constructions, but a natural question arises. Can we construct new codes that have the same or better worst recovery threshold but with a better best recovery threshold? We answer this question in the affirmative. 

To give an improved construction, we shall analyze the existing construction SEP code carefully. Without loss of generality, assume $(K_2+1)(K_1m+X) < (K_1+1)(K_2m+X)$. Thus the encoding polynomials are given as follows:

\begin{align}
p_{\mathbf{A}}(x) & =\sum_{j=1}^{K_1} \sum_{k=1}^{m} \mathbf{A}_{j, k} x^{(k-1)+(j-1)m} + \sum^{X}_{i=1} \mathbf{R}_i x^{K_1m+i-1},\\
p_{\mathbf{B}}(x) & =\sum_{j=1}^{m} \sum_{k=1}^{K_2} \mathbf{B}_{j, k} x^{(m-j)+(K_1m+X)(k-1)} \nonumber \\
& \quad + \sum^{X}_{i=1} \mathbf{T}_i x^{(K_2-1)(K_1m+X)+K_1m+i-1}.
\end{align}

The minimal possible degree that $\mathbf{C}_{i, j}$ appears in the SEP code is $m-1$, thus limiting the performance of the best recovery threshold. By taking the random matrices $\mathbf{R}_i$ into account, there will be a total of $K_1m+X$ matrices involved in the design of encoding polynomials. Thus the minimum possible degree construction for $p_{\mathbf{A}}(x)$ has a degree of $K_1m+X-1$, which involves $K_1m+X$ consecutive degrees. But the summation order of $\mathbf{A}_{j, k}$ may not be canonical in the SEP construction. Note in SEP code, the submatrices of $\mathbf{A}$ are added in the ascending order along the rows. By changing the summation order along the columns will give rise to another encoding polynomial for $\mathbf{A}$ with the same degree. To construct $p_\mathbf{B}(x)$, we need to ensure that $\mathbf{C}_{i, j}$ lies in the coefficients of the product polynomial. Note when $\mathbf{A}_{i, j}$s are encoded along the column, the encoding power corresponding to $\mathbf{A}_{i+1, j}$ will be bigger by $K_1$ than that of $\mathbf{A}_{i, j}$. From the definition of $\mathbf{C}$, the encoding power corresponding to $\mathbf{B}_{i, j}$ will be bigger by $K_1$ than that of $\mathbf{B}_{i, j+1}$. And the random noise terms in $p_{\mathbf{B}}(x)$ shall have consecutive degrees and make the noise term does not interface the needed $\mathbf{C}$ in the product polynomial.

To summarize our construction, recall that we shall admit block partitioning of matrices $\mathbf{A}$ and $\mathbf{B}$ as follows,
\begin{equation*}
 \mathbf{A}=\left[\begin{array}{ccc}
\mathbf{A}_{1,1} & \cdots & \mathbf{A}_{1, m} \\
\vdots & \ddots & \vdots \\
\mathbf{A}_{K_1,1} & \cdots & \mathbf{A}_{K_1, m}
\end{array}\right], 
\end{equation*}
\begin{equation*}
 \mathbf{B}=\left[\begin{array}{ccc}
\mathbf{B}_{1,1} & \cdots & \mathbf{B}_{1, K_2} \\
\vdots & \ddots & \vdots \\
\mathbf{B}_{m,1} & \cdots & \mathbf{B}_{m, K_2}
\end{array}\right].
\end{equation*}

Assume $\mathbf{R}_i$ and $\mathbf{T}_i$ ($1 \leq i \leq X$) are random matrices that are sampled uniformly from $F_q$. Then polynomial-based encoding functions are given by

\begin{align}
p_{\mathbf{A}}(x) & =\sum_{j=1}^{K_1} \sum_{k=1}^{m} \mathbf{A}_{j, k} x^{(k-1)K_1+j-1} + \sum^{X}_{i=1} \mathbf{R}_i x^{K_1m+i-1},\\
p_{\mathbf{B}}(x) & =\sum_{j=1}^{m} \sum_{k=1}^{K_2} \mathbf{B}_{j, k} x^{(m-j)K_1+(K_1m+X)(k-1)} \nonumber \\
& \quad + \sum^{X}_{i=1} \mathbf{T}_i x^{(K_2-1)(K_1m+X)+K_1m+i-1}.
\end{align}

Interestingly, reversing the role of $\mathbf{A}$ and $\mathbf{B}$ shall give another construction.

\begin{align}
p_{\mathbf{A}}(x) & =\sum_{j=1}^{K_1} \sum_{k=1}^{m} \mathbf{A}_{j, k} x^{(k-1)K_2+(K_2m+X)(j-1)} \nonumber\\
& \quad + \sum^{X}_{i=1} \mathbf{R}_i x^{(K_1-1)(K_2m+X)+K_2m+i-1} ,\\
p_{\mathbf{B}}(x) & =\sum_{j=1}^{m} \sum_{k=1}^{K_2} \mathbf{B}_{j, k} x^{(m-j)K_2+k-1} + \sum^{X}_{i=1} \mathbf{T}_i x^{K_2m+i-1}.
\end{align}

\begin{theorem} \label{non DFT transpose SEP}
The proposed scheme is $X$-secure and will have a recovery threshold of $min \{ (K_2+1)(K_1m+X)- 1, (K_1+1)(K_2m+X)- 1 \}$. 
\end{theorem}

\begin{proof}

Without loss of generality, assume $(K_2+1)(K_1m+X) < (K_1+1)(K_2m+X)$, thus the product polynomial $p_{\mathbf{A}}(x)p_{\mathbf{B}}(x)$ is given as follows:

\begin{align}
  &p_{\mathbf{A}}(x)p_{\mathbf{B}}(x) \nonumber \\
  =&  \sum_{j=1}^{K_1} \sum_{k=1}^{m}\sum_{j'=1}^{m} \sum_{k'=1}^{K_2} \nonumber \\
  &\mathbf{A}_{j, k}\mathbf{B}_{j', k'} x^{(k-1)K_1+(j-1)+(m-j')K_1+(k'-1)(K_1m+X)} \label{first term}   \\
  &+ \sum^{X}_{i=1}\sum_{j=1}^{m} \sum_{k=1}^{K_2}  \mathbf{R}_i \mathbf{B}_{j, k} x^{K_1m+i-1+(m-j)K_1+(k-1)(K_1m+X)} \label{second term}  \\
  &+ \sum_{j=1}^{K_1} \sum_{k=1}^{m} \sum^{X}_{i=1} \nonumber \\ 
  &\mathbf{A}_{j, k}\mathbf{T}_i x^{(k-1)K_1+(j-1)+(K_2-1)(K_1m+X)+K_1m+i-1} \label{third term}  \\
  &+ \sum^{X}_{i=1}\sum^{X}_{i'=1} \mathbf{R}_i\mathbf{T}_{i'} x^{K_1m+i-1+(K_2-1)(K_1m+X)+K_1m+i'-1} \label{fourth term}. 
\end{align}   

From the definition of $\mathbf{C}_{j, k^{\prime}}$, we know $\mathbf{C}_{j, k^{\prime}}$ is the coefficient of $x^{(m-1)K_1 + (j-1) + (k^{\prime}-1)(K_1m + X)}$ in equation (\ref{first term}) in the product polynomial. We then show that the powers of x appeared in equations (\ref{second term}), (\ref{third term}) and (\ref{fourth term}) shall not contain the needed information to retrieve $\mathbf{C}_{j, k^{\prime}}$. It is clear that the maximal possible power that $\mathbf{C}_{j, k^{\prime}}$ may appear is $K_1m-1+(K_2-1)(K_1m+X)$. This power is strictly smaller than any of the terms that appeared in equations (\ref{third term}) and (\ref{fourth term}). Thus we shall only consider the terms in equation (\ref{second term}) carefully. 

Note the powers that $\mathbf{C}_{j, k^{\prime}}$ appear are $(m-1)K_1 + (j-1) + (k^{\prime}-1)(K_1m + X)$ ($1 \leq j \leq K_1$, $1 \leq k^{\prime} \leq K_2$). As $0 \leq (m-1)K_1 \leq (m-1)K_1 + (j-1) \leq K_1m-1 < K_1m \leq K_1m+X$, thus we can treat $(m-1)K_1 + (j-1)$ as a remainder term. We then analyze the term in equation $(\ref{second term})$ similarly. The powers of $x$ that appear in the equation $(\ref{second term})$ are $K_1m+i-1+(m-j)K_1+(k-1)(K_1m+X)$ ($1 \leq i \leq X$, $1 \leq j \leq m$ and $1 \leq k \leq K_2$). As $K_1m \leq K_1m+i-1+(m-j)K_1 \leq K_1m+X-1+(m-1)K_1 \leq K_1m+X +(m-1)K_1 -1$. Thus $\mathbf{C}_{j, k^{\prime}}$ lies only in the first term. And the conclusion follows.

For the security concerns, as the degrees of random matrices are consecutive, the security follows from \cite{chang2018capacity}.

\end{proof}

\textbf{Example 2}: ($K_1=m=2, K_2=3, X=1$)

\begin{equation*}
 \mathbf{A}=\left[\begin{array}{cc}
\mathbf{A}_{1,1} &  \mathbf{A}_{1, 2} \\
\mathbf{A}_{2,1} & \mathbf{A}_{2, 2}
\end{array}\right], 
\end{equation*}
\begin{equation*}
 \mathbf{B}=\left[\begin{array}{ccc}
\mathbf{B}_{1,1} &  \mathbf{B}_{1, 2} &  \mathbf{B}_{1, 3} \\
\mathbf{B}_{2,1} & \mathbf{B}_{2, 2} &  \mathbf{B}_{2, 3}
\end{array}\right].
\end{equation*}

Denote the product matrix $\mathbf{C}$ as follows:
\begin{equation*}
  \mathbf{C} =   \left[\begin{array}{ccc}
\mathbf{C}_{1,1} &  \mathbf{C}_{1,2} & \mathbf{C}_{1,3} \\
\mathbf{C}_{2,1} & \mathbf{C}_{2,2} & \mathbf{C}_{2,3}
\end{array}\right],
\end{equation*}
where $\mathbf{C}_{i,j}=\mathbf{A}_{i,1}\mathbf{B}_{1,j}+\mathbf{A}_{i,2}\mathbf{B}_{2,j}$.

The encoding polynomials are as follows:
\begin{align*}
p_{\mathbf{A}}(x)  & = (\mathbf{A}_{1,1}x^3  +\mathbf{A}_{1,2}) +(\mathbf{A}_{2,1}x^{3} +\mathbf{A}_{2,2})x^7 +\mathbf{R}_1 x^{13}. \\
 p_{\mathbf{B}}(x) & = \mathbf{B}_{1,1}+\mathbf{B}_{1,2}x +\mathbf{B}_{1,3}x^{2} \\
& \quad +\mathbf{B}_{2,1}x^{3}+\mathbf{B}_{2,2}x^{4}+\mathbf{B}_{2,3}x^{5}  +\mathbf{T}_1 x^{6}.
\end{align*}

Thus the product $p_{\mathbf{A}}(x)p_{\mathbf{B}}(x)$ shall rewritten as follows,
\begin{align*}
& \mathbf{A}_{1,2}(\mathbf{B}_{1,1}+\mathbf{B}_{1,2}x+\mathbf{B}_{1,3}x^2) + (\mathbf{A}_{1,1}\mathbf{B}_{1,1}+\mathbf{A}_{1,2}\mathbf{B}_{2,1})x^3 + \\& (\mathbf{A}_{1,1}\mathbf{B}_{1,2}+\mathbf{A}_{1,2}\mathbf{B}_{2,2})x^4 
+(\mathbf{A}_{1,1}\mathbf{B}_{1,3}+\mathbf{A}_{2,1}\mathbf{B}_{2,3})x^5 + \cdots  \\
 &
 +(\mathbf{A}_{2,1}\mathbf{B}_{1,1} + \mathbf{A}_{2,2}\mathbf{B}_{2,1})x^{10} + (\mathbf{A}_{2,1}\mathbf{B}_{1,2} + \mathbf{A}_{2,2}\mathbf{B}_{2,2})x^{11} \\
& + (\mathbf{A}_{2,1}\mathbf{B}_{1,3}+\mathbf{A}_{2,2}\mathbf{B}_{2,3})x^{12} + \cdots + \mathbf{R}_1 \mathbf{T}_1 x^{19}.
\end{align*}

The desired power is $3,4,5,10,11,12$. Thus interpolating on the $17$-th root of unity set will be sufficient to obtain the best recovery threshold. And note the product polynomial has a degree of $19$, thus the worst recovery threshold is $20$.

\begin{theorem} \label{reduce SEP theorem}
Denote $\hat A_1 = (K_2+1)(K_1m+X)- 1$ and $\hat A_2 = (K_1+1)(K_2m+X)- 1$. Without loss of generality, assume $\hat A_1 < \hat A_2$. When $N \geq \hat A_1$, assume the $\hat A-K_1(m-1)$-th root of unity set is a subset of the evaluation set. Then the worst recovery threshold is $\hat A_1$, and the best recovery threshold is $\hat A_1-K_1(m-1)$.
    
\end{theorem}

\begin{proof}
The worst recovery threshold comes from the derivation in theorem \ref{non DFT transpose SEP}. As discussed in theorem \ref{non DFT transpose SEP}, all the sub-matrices needed to retrieve the product matrix $\mathbf{C}$ lie in the coefficients in equation (\ref{first term}), but the first $K_1(m-1)$ coefficients are not needed. As the noise term in equation (\ref{fourth term}) contains exactly the highest powers in the product polynomial. As the top-$K_1(m-1)$ powers leave the powers of $\mathbf{C}_{j, k^{\prime}}$ untouched, we can successfully apply the modulo technique to obtain the best recovery threshold. And the security follows from \cite{chang2018capacity}.    
\end{proof}

\textbf{Remark 3}: Note that from theorem \ref{reduce SEP theorem}, our construction is strictly better than SEP code in the sense that it has the same worst recovery threshold but with a better best recovery threshold.

\subsection{Apply Modulo Technique to Entangled Polynomial Code} \label{EP}

Though entangled polynomial code can be seen as taking $X=0$ in secure entangled polynomial code, directly applying the techniques in \ref{SEP} is not efficient. We shall propose a construction that can group the interpolation points to give much flexible recovery for entangled polynomial code.

\begin{theorem} \label{EP theorem}
    When $N \geq K_1K_2m+m-1$, assume the $m\lfloor \frac{N}{m} \rfloor$-th root of unity set is a subset of the evaluation set. Then the worst recovery threshold is $K_1K_2m+m-1$, the best recovery threshold is $K_1K_2m$.

\end{theorem}

\begin{proof}
As the encoding functions are the same as the entangled polynomial code, thus having a worst recovery threshold of $K_1K_2m+m-1$. 
 
We shall next present the product polynomial $p_{\mathbf{A}}(x)p_{\mathbf{B}}(x)$ of entangled polynomial code.
 
\begin{align}
  &p_{\mathbf{A}}(x)p_{\mathbf{B}}(x) \nonumber  \\
  &=\hspace{-1mm} \sum_{j=1}^{K_1}\hspace{-0.7mm} \sum_{k=1}^{m} \hspace{-0.7mm}\sum_{j'=1}^{m} \hspace{-0.7mm} \sum_{k'=1}^{K_2} \hspace{-0.7mm}\mathbf{A}_{j, k}\mathbf{B}_{j', k'} x^{(k\hspace{-0.6mm}-\hspace{-0.6mm}1)+(j\hspace{-0.6mm}-\hspace{-0.6mm}1)m+(m\hspace{-0.6mm}-\hspace{-0.6mm}j')+(k'\hspace{-0.6mm}-\hspace{-0.6mm}1)K_1m}.  
\end{align}

As $\mathbf{C}_{i, j} = \sum_{k=1}^{m} \mathbf{A}_{i, k} \mathbf{B}_{k, j}$, it is the coefficient of the $m-1+(i-1)m+(j-1)K_1m$-th degree term. The coefficients of $x^{j m-1}$ ($j \in\{1, \cdots, K_1K_2\}$) are what we need. We can group evaluation points if their $m$-th power is equal, then there will be $\lfloor \frac{N}{m} \rfloor$ groups. If we take $A(x)$ in lemma \ref{lemma 2} as the product matrix polynomial $p_{\mathbf{A}}(x)p_{\mathbf{B}}(x)$ and evaluate it on each group, it is interesting to find that $B(x)$ is exactly the product matrix polynomial of Polynomial code-like partitioning of the matrices. By the recovery threshold of Polynomial code, once the computation results of any $K_1K_2$ groups are returned to the master, the master can also decode the matrix product. Thus giving the best recovery threshold of $K_1K_2m$.
\end{proof}

\textbf{Example 3}: ($K_1=K_2=m=2, N=10$)

\begin{equation*}
 \mathbf{A}=\left[\begin{array}{cc}
\mathbf{A}_{1,1} &  \mathbf{A}_{1, 2} \\
\mathbf{A}_{2,1} & \mathbf{A}_{2, 2}
\end{array}\right], 
\end{equation*}
\begin{equation*}
 \mathbf{B}=\left[\begin{array}{cc}
\mathbf{B}_{1,1} &  \mathbf{B}_{1, 2} \\
\mathbf{B}_{2,1} & \mathbf{B}_{2, 2}
\end{array}\right]
\end{equation*}

It is clear that 
\begin{equation*}
  \mathbf{A}\mathbf{B} =   \left[\begin{array}{cc}
\mathbf{A}_{1,1}\mathbf{B}_{1,1}+\mathbf{A}_{1,2}\mathbf{B}_{2,1} &  \mathbf{A}_{1,2}\mathbf{B}_{2,2}+\mathbf{A}_{1,1}\mathbf{B}_{1,2} \\
\mathbf{A}_{2,1}\mathbf{B}_{1,1}+\mathbf{A}_{2,2}\mathbf{B}_{2,1} & \mathbf{A}_{2,1}\mathbf{B}_{1,2}+\mathbf{A}_{2,2}\mathbf{B}_{2,2}
\end{array}\right].
\end{equation*}

Denote $\zeta$ as the $10$-th primitive root of unity. Then $\zeta^2$ is the $5$-th primitive root of unity and $\zeta^5=-1$. Let $x_i=\zeta^i$.

Then encode matrices are 
$$
\begin{aligned}
\tilde{\mathbf{A}}_i & =\mathbf{A}_{1,1}+\mathbf{A}_{1,2}x_{i}+\mathbf{A}_{2,1}x^{2}_{i}+\mathbf{A}_{2,2}x^{3}_{i}, \\
\tilde{\mathbf{B}}_i & =\mathbf{B}_{1,1}x_i+\mathbf{B}_{2,1}+\mathbf{B}_{1,2}x^{5}_{i}+\mathbf{B}_{2,2}x^{4}_{i}.
\end{aligned}
$$

Thus the product $\tilde{\mathbf{A}}_i\tilde{\mathbf{B}}_i$ shall rewritten as follows,

\begin{align*}
& \mathbf{A}_{1,1}\mathbf{B}_{2,1} +(\mathbf{A}_{1,2}\mathbf{B}_{1,1}+\mathbf{A}_{2,1}\mathbf{B}_{2,1})x^2_i +\mathbf{A}_{2,2}\mathbf{B}_{1,2}x^8_i \\& +(\mathbf{A}_{2,2}\mathbf{B}_{1,1}+\mathbf{A}_{1,1}\mathbf{B}_{2,2})x^4_i
+(\mathbf{A}_{2,1}\mathbf{B}_{2,2}+\mathbf{A}_{1,2}\mathbf{B}_{1,2})x^6_i   \\
 &+ (\mathbf{A}_{1,1}\mathbf{B}_{1,1}+\mathbf{A}_{1,2}\mathbf{B}_{2,1})x_i
 +(\mathbf{A}_{2,1}\mathbf{B}_{1,1}+\mathbf{A}_{2,2}\mathbf{B}_{2,1})x^3_i \\
&+(\mathbf{A}_{1,2}\mathbf{B}_{2,2}+\mathbf{A}_{1,1}\mathbf{B}_{1,2})x^5_i
 +(\mathbf{A}_{2,1}\mathbf{B}_{1,2}+\mathbf{A}_{2,2}\mathbf{B}_{2,2})x^7_i.
\end{align*}

As the product polynomial is of degree $8$, any $9$ workers' computational responses can recover the product matrix $\mathbf{A}\mathbf{B}$.

It is clear that for any $1 \leq i \leq 5$,
\begin{align*}
     &\frac{\tilde{\mathbf{A}}_i\tilde{\mathbf{B}}_i-\tilde{\mathbf{A}}_{5+i}\tilde{\mathbf{B}}_{5+i}}{\zeta^i-\zeta^{5+i}} \\ 
      =  
 &(\mathbf{A}_{1,1}\mathbf{B}_{1,1}+\mathbf{A}_{1,2}\mathbf{B}_{2,1})\frac{\zeta^i-\zeta^{5+i}}{\zeta^i-\zeta^{5+i}}\\ &+(\mathbf{A}_{2,1}\mathbf{B}_{1,1}+\mathbf{A}_{2,2}\mathbf{B}_{2,1})\frac{(\zeta^2)^i(\zeta^i-\zeta^{5+i})}{\zeta^i-\zeta^{5+i}} \\
&+(\mathbf{A}_{1,2}\mathbf{B}_{2,2}+\mathbf{A}_{1,1}\mathbf{B}_{1,2})\frac{((\zeta^2)^i)^2(\zeta^i-\zeta^{5+i})}{\zeta^i-\zeta^{5+i}}  \\
& +(\mathbf{A}_{2,1}\mathbf{B}_{1,2}+\mathbf{A}_{2,2}\mathbf{B}_{2,2})\frac{((\zeta^2)^i)^3(\zeta^i-\zeta^{5+i})}{\zeta^i-\zeta^{5+i}}. 
\end{align*}

If we define a new polynomial $h(x)$ as \begin{align*}
    &(\mathbf{A}_{1,1}\mathbf{B}_{1,1}+\mathbf{A}_{1,2}\mathbf{B}_{2,1}) +(\mathbf{A}_{2,1}\mathbf{B}_{1,1}+\mathbf{A}_{2,2}\mathbf{B}_{2,1})x + \\
& (\mathbf{A}_{1,2}\mathbf{B}_{2,2}+\mathbf{A}_{1,1}\mathbf{B}_{1,2})x^2  +(\mathbf{A}_{2,1}\mathbf{B}_{1,2}+\mathbf{A}_{2,2}\mathbf{B}_{2,2})x^3. 
\end{align*}

Then $\frac{\tilde{\mathbf{A}}_i\tilde{\mathbf{B}}_i-\tilde{\mathbf{A}}_{5+i}\tilde{\mathbf{B}}_{5+i}}{\zeta^i-\zeta^{4+i}} = h((\zeta^2)^i)$.

As $h(x)$ is of degree $3$, any $4$ evaluation results are sufficient to decode. Thus we can get the coefficient easily, which is exactly the matrix product $\mathbf{A}\mathbf{B}$.

\textbf{Remark 4}: The techniques in \ref{SEP} and \ref{EP} can also be applied to the $n$ matrix code proposed in \cite{dutta2019optimal} to reduce its recovery threshold.

\subsection{A New Distributed Matrix Multiplication Scheme Based on Locally Repairable Codes}

Assume matrix $\mathbf{A}$ is split vertically into $m$ equal column-blocks and $\mathbf{B}$ is split horizontally into $m$ equal row blocks as follows:
$$
\mathbf{A}=\left[\begin{array}{llll}
\mathbf{A}_{1} & \mathbf{A}_{2} & \ldots & \mathbf{A}_{m}
\end{array}\right], \quad \mathbf{B}=\left[\begin{array}{c}
\mathbf{B}_{1} \\
\mathbf{B}_{2} \\
\vdots \\
\mathbf{B}_{m}
\end{array}\right]
.$$

Recently, \cite{SDMMDFT} proposes a discrete fourier transform based code. Their encoding functions are as follows:
\begin{equation*}
 p_{\mathbf{A}}(x)= \sum^{m-1}_{i=0}\mathbf{A}_{i+1}x^i, 
p_{\mathbf{B}}(x)= \sum^{m-1}_{i=0}\mathbf{B}_{i+1}x^{-i}.
\end{equation*}

Suppose the $m$-th primitive root of unity is $\zeta$. There will be a total of $m$ workers, each worker $i$ will be sent the evaluation results of $p_{\mathbf{A}}$ and $p_{\mathbf{B}}$ at $\zeta^i$. Note $p_{\mathbf{A}}(x)p_{\mathbf{B}}(x)=\mathbf{A}\mathbf{B}+ (\text{non constant terms})$. Using the property of root of unity, once the master receives all the computation results, it shall use $\frac{\sum^{m}_{i=1}p_{\mathbf{A}}(\zeta^i)p_{\mathbf{B}}(\zeta^i)}{m}$ to retrieve the matrix product. It achieves a recovery threshold of $m$, while MatDot code has a recovery threshold of $2m-1$. But discrete fourier transform based code cannot mitigate stragglers. Recall that locality describes each codeword symbol may be reconstructed by some other symbols \cite{prakash2012optimal}, Thus we would expect to construct code that exploits locality to resist stragglers while still retaining a good recovery threshold.

We shall introduce the construction of our scheme.

Suppose $r$ can take any value between $1$ and $2m-1$. Let $\mathcal{A}=\left\{\alpha_{1}, \cdots \alpha_{N}\right\}$ be a set of $N$ distinct numbers and let $\mathcal{A}_{1}, \cdots, \mathcal{A}_{\frac{N}{r+\delta-1}}$ be subsets of $\mathcal{A}$ with size $(r+\delta-1)$ which form a partition of $\mathcal{A}$. For simplicity, we assume that $r+\delta-1$ divides both $N$ and $m$. Let $g(x)$ be a polynomial of degree $r+\delta-1$ which is constant on each subset $\mathcal{A}_{i}$. Inspired by \cite{jeong2018locally}, encoding polynomials of the matrices $\mathbf{A}$ and $\mathbf{B}$ are as follows:

\begin{align}
p_{\mathbf{A}}(x) &=\left(\mathbf{A}_{1}+\cdots+\mathbf{A}_{\frac{r+1}{2}} x^{\frac{r-1}{2}}\right) \nonumber \\
&+\left(\mathbf{A}_{\frac{r+1}{2}+1}+\cdots+\mathbf{A}_{(r+1)} x^{\frac{r-1}{2}}\right) g(x)+\cdots \nonumber \\
&+\left(\mathbf{A}_{m-\frac{r+1}{2}+1}+\cdots+\mathbf{A}_{m} x^{\frac{r-1}{2}}\right) g(x)^{\frac{2 m}{r+1}-1}, \\
p_{\mathbf{B}}(x) &=\left(\mathbf{B}_{m}+\cdots+\mathbf{B}_{m-\frac{r+1}{2}+1} x^{\frac{r-1}{2}}\right) \nonumber \\
&+\left(\mathbf{B}_{m-\frac{r+1}{2}}+\cdots+\mathbf{B}_{m-r} x^{\frac{r-1}{2}}\right) g(x)+\cdots \nonumber \\
&+\left(\mathbf{B}_{\frac{r+1}{2}}+\cdots+\mathbf{B}_{1} x^{\frac{r-1}{2}}\right) g(x)^{\frac{2 m}{r+1}-1} .
\end{align}

\begin{theorem} \label{LRC recover thre}
The construction given above
achieves a recovery threshold $K = 4m-r-2+(\frac{4m}{r+1}-2)(\delta-2)$. 
\end{theorem}

\begin{proof}
First, note that we have the product polynomial written as follows. 
$$
\begin{aligned}
p_{\mathbf{C}}(x) &=p_{\mathbf{A}}(x) p_{\mathbf{B}}(x) \\
&=\sum_{j=1}^{\frac{r+1}{2}} \mathbf{A}_{j} x^{j-1} \sum_{j=1}^{\frac{r+1}{2}} \mathbf{B}_{m-j+1} x^{j-1}+\cdots \\
&+\left(\mathbf{A}_{1} \mathbf{B}_{1}+\cdots+\mathbf{A}_{m} \mathbf{B}_{m}\right) x^{\frac{r-1}{2}} g(x)^{\frac{2 m}{r+1}-1}+\cdots \\
&+\sum_{j=1}^{\frac{r+1}{2}} \mathbf{A}_{m-j+1} x^{j-1} \sum_{j=1}^{\frac{r+1}{2}} \mathbf{B}_{j} x^{j-1} g(x)^{\frac{4 m}{r+1}-2}.
\end{aligned}
$$
It is easy to find that the coefficient of $x^{\frac{r-1}{2}} g(x)^{\frac{2 m}{r+1}-1}$ is exactly the matrix product $\mathbf{A}\mathbf{B}$.
The degree of polynomial has the relationship, $deg( p_{\mathbf{C}}(x)) =deg( p_{\mathbf{A}}(x)) + deg( p_{\mathbf{B}}(x))=4m-r-3+(\frac{4m}{r+1}-2)(\delta-2)$. Thus getting a recovery threshold $K=deg(p_{\mathbf{C}}(x))+1=4m-r-2+(\frac{4m}{r+1}-2)(\delta-2)$. 

We shall then analysis the locality of the code. When evaluated on each local group $\mathcal{A}_i$, $g(x)$ remains a constant $\gamma_i$. Thus from the expression of $p_{\mathbf{A}}(x)$ and $p_{\mathbf{B}}(x)$, they both have degree $\frac{r-1}{2}$ and the degree of product polynomial $p_{\mathbf{C}}(x)$ will be $r-1$. Therefore, any $r$ computation results on $\mathcal{A}_i$ will be sufficient to decode the product polynomial, then getting all the evaluation results on $\mathcal{A}_i$. 
\end{proof}

\begin{theorem}
\label{LRC}
Assume that $\frac{r+1}{2} | m$. Then when interpolating at the $\frac{m}{\frac{r+1}{2}}(r+\delta-1)$-th root of unity set, we can benefit from the local group repair structure and the modulo trick.
\end{theorem}

\begin{proof}

We only need the coefficient of $x^{\frac{r-1}{2}}g(x)^{\frac{2m}{r+1}-1}$, whose degree is half of that of product polynomial. Thus it is suitable to apply lemma \ref{lemma 1}. We want to choose a suitable degree $k$ that when interpolating on the $k$-th root of unity set ensures a lower recovery threshold and locally repairable structure at the same time. In order to maintain the requirements of lemma \ref{lemma 1}, we need $$deg(x^{\frac{r-1}{2}}g(x)^{\frac{2m}{r+1}-1})>deg(p_{\mathbf{C}}(x)) - k$$ and $$k-1 \geq deg(x^{\frac{r-1}{2}}g(x)^{\frac{2m}{r+1}-1}).$$ Thus we need $k \geq deg(x^{\frac{r-1}{2}}g(x)^{\frac{2m}{r+1}-1})+1$. From the division requirement of local group, we shall need $r + \delta -1 | k$. Note that $deg(x^{\frac{r-1}{2}}g(x)^{\frac{2m}{r+1}-1})+1 = 2m-\frac{r+1}{2}+(\frac{2m}{r+1}-1)(r+\delta-1)-(r+1)(\frac{2m}{r+1}-1)=\frac{r+1}{2} + (\frac{2m}{r+1}-1)(r+\delta-1)$. We also know $\frac{r+1}{2} \leq r + \delta -1$, the smallest $k$ satisfies these is $(r+\delta-1)\frac{2m}{r+1}$.
\end{proof}

\begin{corollary} \label{th1 corr}
When applying the above scheme, we shall have the best recovery threshold $2m-\frac{m}{\frac{r+1}{2}}$ and worst recovery threshold $2m-1+(\frac{m}{\frac{r+1}{2}}-1)(\delta-2)$.
\end{corollary}

\begin{proof}
The best recovery threshold is achieved whenever we receive any $r$ responses from all the local groups $\mathbf{A}_i$. The worst recovery threshold is achieved whenever the slowest workers are all from the same local group, which we can only resist $\delta-1$ of them.  
\end{proof}

\textbf{Remark 5}: Note the LRC MatDot code in \cite{jeong2018locally} is a $(r, \delta = 2)$ case of our scheme and MatDot Code is a $(r=1, \delta=1)$ case of our scheme.

\textbf{Example 4}: ($m=6,  r=3, \delta=3, N=15$ )

Let $g(x)=x^5$, a polynomial of degree $5$.

The encoding polynomials are as follows:
\begin{align*} 
& p_{\mathbf{A}}(x) \hspace{-1mm} = \hspace{-1mm}\left(\mathbf{A}_1+\mathbf{A}_2 x\right)\hspace{-1mm}  +\hspace{-1mm}\left(\mathbf{A}_3+\mathbf{A}_4 x\right) g(x)\hspace{-1mm}+\hspace{-1mm}\left(\mathbf{A}_5+\mathbf{A}_6 x\right) g(x)^2. \\
& p_{\mathbf{B}}(x) \hspace{-1mm} = \hspace{-1mm} \left(\mathbf{B}_6+\mathbf{B}_5 x\right)\hspace{-1mm}+\hspace{-1mm}\left(\mathbf{B}_4+\mathbf{B}_3 x\right) g(x)\hspace{-1mm}+\hspace{-1mm}\left(\mathbf{B}_2+\mathbf{B}_1 x\right) g(x)^2.
\end{align*}
As the degree of product polynomial $p_{\mathbf{C}}(x)$ is $22$ and $\mathbf{A}\mathbf{B}$ is the coefficient of $x^{11}$ in the product polynomial. We shall choose any degree $k \geq 12$ to satisfy the requirement of lemma \ref{lemma 1}. Note each local group is of size $5$, thus we want $5|k$. The minimum choice of $k$ is $15$. 

Denote $\zeta$ be the $15$-th primitive root, then 
\begin{equation*}
\mathcal{A}=\left\{\zeta^{1}, \cdots, \zeta^{15}\right\}    
\end{equation*}
will be the set of $15$-th root of unity. Then letting
\begin{equation*}
\mathcal{A}_{1}=\left\{\zeta^{1}, \zeta^{4},\cdots, \zeta^{13}\right\}, \cdots, \mathcal{A}_{3}=\left\{\zeta^{3}, \zeta^{6},\cdots, \zeta^{15}\right\} 
\end{equation*}
be subsets of $\mathcal{A}$ of size $5$ that form a partition of $\mathcal{A}$. Thus $g(x)$ is constant on each subset $\mathcal{A}_{i}$.

The $i$-th worker node receives the encoded matrices, $p_{\mathrm{A}}\left(\zeta^i\right)$ and $p_{\mathbf{B}}\left(\zeta^i\right)$, and then computes the following product:
$$
\begin{aligned}
p_{\mathbf{C}}(x)= & \left(\mathbf{A}_1+\mathbf{A}_2 x\right)\left(\mathbf{B}_6+\mathbf{B}_5 x\right)+\cdots \\
& +\left(\mathbf{A}_1 \mathbf{B}_1+\cdots \mathbf{A}_6 \mathbf{B}_6\right) x g(x)^2+\cdots \\
& +\left(\mathbf{A}_5+\mathbf{A}_6 x\right)\left(\mathbf{B}_2+\mathbf{B}_1 x\right) g(x)^4 .
\end{aligned}
$$
at $x=\zeta^i$.

Let us denote $g\left(\mathcal{A}_1\right)=\gamma$. Then $p_{\mathbf{C}}(x)$ at $\zeta^1, \zeta^4, \cdots, \zeta^{13}$ can be rewritten as:
$$
\begin{aligned}
p_{\mathbf{C}}(x)= & \left(\mathbf{A}_1+\mathbf{A}_2 x\right)\left(\mathbf{B}_6+\mathbf{B}_5 x\right)+\cdots \\
& +\left(\mathbf{A}_1 \mathbf{B}_1+\cdots \mathbf{A}_6 \mathbf{B}_6\right) x \gamma^2+\cdots \\
& +\left(\mathbf{A}_5+\mathbf{A}_6 x\right)\left(\mathbf{B}_2+\mathbf{B}_1 x\right) \gamma^4.
\end{aligned}
$$

Now, notice that this is a polynomial of degree $2$, which can be recovered from evaluation at any three points. The locality may be used to give a flexible recovery threshold. The best case is that we can get any $3$ evaluation results in each $\mathcal{A}_{i}$, then using locality we can recover all the interpolation results and thus recover the matrix product, giving a recovery threshold of $9$. The worst case is that we get $13$ computation results and the $2$ unfinished results are from some same $\mathcal{A}_{i}$.

We shall give a construction of LRC MatDot Code that can also maintain security when $X \leq \frac{r+1}{2}$. 

The partition of matrices and the configuration of polynomial $g(x)$ with local groups are similar as previously discussed. Matrices $\mathbf{R}_{1}, \ldots, \mathbf{R}_{X} $ and $\mathbf{T}_{1}, \ldots, \mathbf{T}_{X} $ are drawn at random. The encoding polynomials of the matrices $\mathbf{A}$ and $\mathbf{B}$ are given as follows:

\begin{align}
p_{\mathbf{A}}(x) &=\left(\mathbf{A}_{1}+\cdots+\mathbf{A}_{\frac{r+1}{2}} x^{\frac{r-1}{2}}\right) \nonumber \\
&+\left(\mathbf{A}_{\frac{r+1}{2}+1}+\cdots+\mathbf{A}_{(r+1)} x^{\frac{r-1}{2}}\right) g(x)+\cdots \nonumber \\
&+\left(\mathbf{A}_{m-\frac{r+1}{2}+1}+\cdots+\mathbf{A}_{m} x^{\frac{r-1}{2}}\right) g(x)^{\frac{2 m}{r+1}-1} \nonumber \\
&+(\mathbf{R}_{1}+\cdots+\mathbf{R}_{X} x^{X-1}) g(x)^{\frac{2 m}{r+1}} , \\
p_{\mathbf{B}}(x) &=\left(\mathbf{B}_{m}+\cdots+\mathbf{B}_{m-\frac{r+1}{2}+1} x^{\frac{r-1}{2}}\right) \nonumber \\
&+\left(\mathbf{B}_{m-\frac{r+1}{2}}+\cdots+\mathbf{B}_{m-r} x^{\frac{r-1}{2}}\right) g(x)+\cdots \nonumber \\
&+\left(\mathbf{B}_{\frac{r+1}{2}}+\cdots+\mathbf{B}_{1} x^{\frac{r-1}{2}}\right) g(x)^{\frac{2 m}{r+1}-1} \nonumber \\
&+(\mathbf{T}_{1}+\cdots+\mathbf{T}_{X} x^{X-1}) g(x)^{\frac{2 m}{r+1}}.
\end{align}

\begin{theorem} \label{th2 secure}
The construction given above is $X$-secure. It has recovery threshold $K = 4m+2\delta+2X-5+(\frac{4m}{r+1}-2)(\delta-2)$. 
\end{theorem}

\begin{proof}
The validity of security comes from \cite{chang2018capacity}. It is clear that coefficient of $x^{\frac{r-1}{2}} g(x)^{\frac{2 m}{r+1}-1}$ of product polynomial is exactly the matrix product $\mathbf{A}\mathbf{B}$. Calculating the degree of product polynomial will imply the recovery threshold. 
\end{proof}

\section{Performance Comparisons} \label{sec:compar}

The state-of-the-art methods are reflected in secure entangled polynomial code (SEP) \cite{zhu2021improved}, polynomial sharing (PS) code \cite{nodehi2018entangled}, discrete fourier transform (DFT) based code \cite{SDMMDFT}, entangled polynomial code (EP) \cite{yu2020straggler}. We name the proposed methods as SEP-DFT, EP-DFT and LRC-DFT.

Table \ref{comparison1} shows the results about block partitioned schemes for distributed matrix matrix multiplication problems. We can see that our method can consistently improve the recovery threshold with a number of $m-1$.

Table \ref{comparison2} shows the results about inner product partitioned schemes for distributed matrix matrix multiplication problems. As DFT code has small recovery threshold, it cannot resist straggler. MatDot code on the other hand can resist straggler, but with a much bigger recovery threshold. LRC-DFT code can offer a flexible trade-off between the recovery threshold and the number of allowed stragglers and subsumes DFT code and LRC MatDot code \cite{jeong2018locally} as special cases.

\begin{table}[h]
     \caption{Comparison of different distributed matrix multiplication schemes based on block partitioning, where the job is to compute the product of the matrices $\mathbf{A}$ and $\mathbf{B}$, $\mathbf{A}$ is divided into $K_1$ block rows and $m$ block columns, $\mathbf{B}$ is divided into $m$ block rows and $K_2$ block columns, $\hat A_1 = (K_2+1)(K_1m+X)- 1$, $\hat A_2 = (K_1+1)(K_2m+X)- 1$, $\hat A_3 = 2K_1K_2m+2X-1$ and $\hat A = min \{ \hat A_1, \hat A_2, \hat A_3 \}$.}
    \centering
    \begin{tabular}{|l|l|}
    \hline
        Method & Recovery Threshold  \\ \hline
        SEP and PS \cite{zhu2021improved,nodehi2018entangled} & $\hat A$  \\ \hline
        EP \cite{yu2020straggler }& $K_1K_2m+m-1$ \\ \hline
        SEP-DFT (from Theorem \ref{SEP theorem}) &  $\hat A-m+1 \sim \hat A$ \\ \hline
        The proposed code in Theorem \ref{reduce SEP theorem} &  $\hat A_1 - K_1(m-1) \sim \hat A_1$ \text{or}  \\ 
        & $\hat A_2-K_2(m-1) \sim \hat A_2$ \\ \hline
        EP-DFT (from Theorem \ref{EP theorem}) & $K_1K_2m \sim K_1K_2m + m - 1$  \\ \hline
    \end{tabular}

    \label{comparison1}
\end{table}

\begin{table}[h]
    \caption{Comparison of different distributed matrix multiplication schemes based on inner product partitioning, where the job is to compute the product of the matrices $\mathbf{A}$ and $\mathbf{B}$, $\mathbf{A}$ is divided into $m$ block columns, $\mathbf{B}$ is divided into $m$ block rows, $A_1=2m-\frac{m}{\frac{r+1}{2}}$, $A_2=2m-1+(\frac{m}{\frac{r+1}{2}}-1)(\delta-2)$.}
    \centering
    \begin{tabular}{|l|l|}
    \hline
        Method & Recovery Threshold  \\ \hline
        DFT \cite{SDMMDFT} & $m$   \\ \hline
        MatDot \cite{dutta2019optimal} & $2m-1$  \\ \hline
        LRC-DFT (from Theorem \ref{LRC}) & $A_1 \sim A_2$ \\ \hline
    \end{tabular}

    \label{comparison2}

\end{table}

\section{Applying to matrix computing with special structure} \label{structure mat mult}

We find many problems like generalized linear models (GLM) \cite{data2020data}, singular value decomposition (SVD) \cite{bentbib2015block}, matrix factorization via alternating least squares \cite{wang2021coded}, and optimal transport (OT) \cite{cuturi2013sinkhorn} share a similar structures. Specifically, there will be a fixed matrix $\mathbf{A}$ and two matrices $\mathbf{U}^{(\ell)}$ and $\mathbf{V}^{(\ell+1)}$ updating during the iterative steps. The main bottleneck will be two products $\mathbf{A}\mathbf{U}^{(\ell)}$ and $\mathbf{A}^T \mathbf{V}^{(\ell+1)}$. Assume that $\mathbf{A}$ is a $a$ by $b$ matrix, $\mathbf{U}^{(\ell)}$ and $\mathbf{V}^{(\ell+1)}$ both have $c$ columns.

To alleviate the expensive bottleneck, we will apply the coding theoretic techniques outlined in the preceding sections. Our suggested method is designed to be memory-efficient and capable of ensuring double-sided security while detecting byzantine workers. Under this approach, we will assume that $\mathbf{A}$, $\mathbf{U}^{(\ell)}$, and $\mathbf{V}^{(\ell+1)}$ will be partitioned into blocks of size $K_1 \times m$, $m \times s$, and $K_1 \times s$, respectively. 

Without loss of generality, assume $(K_2+1)(K_1m+X) < (K_1+1)(K_2m+X)$.

The proposed method works as follows:

\textbf{Preliminary Step}
Choose $N$ distinct numbers $x_i$ ($i = 1,2, \cdots, N$), where $x_i$s contain the $(K_2+1)(K_1m+X)-K_2(m-1)$-th root of unity set. Denote  
\begin{equation} \label{A encode}
\bar{\mathbf{A}}_i  =\sum_{j=1}^{K_1} \sum_{k=1}^{m} \mathbf{A}_{j, k} x^{(k-1)+(j-1)m}_i + \sum^{X}_{i=1} \mathbf{R}_i x^{K_1m+i-1}_i ,  
\end{equation}
then the master distributes the encoded matrix $\bar{\mathbf{A}}_i$ to worker $i$ ($i \in [N-1]$). 

\textbf{Checking Key Generation} The master will generate two verification row vector keys $\mathbf{r}^{1}_i \in F^{1 \times \frac{a}{K_1}}_q, \mathbf{r}^{2}_i \in F^{1 \times \frac{b}{m}}_q$ for each worker $i$, where each coordinate of $\mathbf{r}^{1}_i, \mathbf{r}^{2}_i$ is from the uniform distribution on $F_q$. The master then creates two vectors $\mathbf{s}^{1}_i \in F^{1 \times \frac{b}{m}}_q, \mathbf{s}^{2}_i \in F^{1 \times \frac{a}{K_1}}_q$ for each worker $i$, where 
\begin{equation} \label{check 1}
\mathbf{s}^{1}_i= \mathbf{r}^{1}_i \bar{\mathbf{A}}_i    
\end{equation}
and 
\begin{equation} \label{check 2}
\mathbf{s}^{2}_i= \mathbf{r}^{2}_i  \bar{\mathbf{A}}^T_i. 
\end{equation}
The master will maintain these vectors $\mathbf{r}^{1}_i, \mathbf{r}^{2}_i, \mathbf{s}^{1}_i, \mathbf{s}^{2}_i$ during the whole process.

\textbf{Then we shall get into the details of the iterative process, we shall use $\ell$ to denote the $\ell$-th iteration.}

Suppose 
\begin{equation}
 \mathbf{U}^{(\ell)}=\left[\begin{array}{ccc}
\mathbf{U}^{(\ell)}_{1,1} & \cdots & \mathbf{U}^{(\ell)}_{1, s} \\
\vdots & \ddots & \vdots \\
\mathbf{U}^{(\ell)}_{m,1} & \cdots & \mathbf{U}^{(\ell)}_{m, s}
\end{array}\right],
\end{equation} and 

\begin{equation}
 \mathbf{V}^{(\ell+1)}=\left[\begin{array}{ccc}
\mathbf{V}^{(\ell+1)}_{1,1} & \cdots & \mathbf{V}^{(\ell+1)}_{1, s} \\
\vdots & \ddots & \vdots \\
\mathbf{V}^{(\ell+1)}_{K_1,1} & \cdots & \mathbf{V}^{(\ell+1)}_{K_1, s}
\end{array}\right].
\end{equation}

\textbf{The Computation Phase I}
The master encodes $\mathbf{U}^{(\ell)}$ into 
\begin{align}\label{U encode}
\bar{\mathbf{U}}^{(\ell)}_i &= \sum_{j=1}^{m} \sum_{k=1}^{s} \mathbf{U}^{(\ell)}_{j, k} x^{(m-j)+(k-1)(K_1m+X)}_i \nonumber\\ &+ \sum^{X}_{i=1} \mathbf{T}_i x^{(s-1)(K_1m+X)+K_1m+i-1}_i,
\end{align}
and distributes the encoded matrix $\bar{\mathbf{U}}^{(\ell)}_i$ to worker $i$. Each worker $i$ calculates $\bar{\mathbf{A}}_i \bar{\mathbf{U}}^{(\ell)}_i$ and returns it to the master.

\textbf{Checking and Decoding Phase I}

Note that what each worker does is just matrix multiplication, thus it suits Freivalds’ algorithm well \cite{freivalds1979fast}. Specifically, whenever the master receives the computation result $\mathbf{\lambda}^{(\ell)}_i$ from the worker $i$, it will do the following checking 
\begin{equation} \label{verify 1}
\mathbf{r}^{1}_i \mathbf{\lambda}^{(\ell)}_i = \mathbf{s}^{1}_i \bar{\mathbf{U}}^{(\ell)}_i.  
\end{equation}
If the equation holds, this worker passes the byzantine check. When the master receives enough results from non-byzantine workers, it will decode $\mathbf{A} \mathbf{U}^{(\ell)}$. Then the master will then transform the matrix multiplication result $\mathbf{A} \mathbf{U}^{(\ell)}$ into $\mathbf{V}^{(\ell +1)}$.

\textbf{The Computation Phase II}
The master encodes $\mathbf{V}^{(\ell +1)}$ using the construction in section \ref{transpose sep code}, i.e. 
\begin{align}
\bar{\mathbf{V}}^{(\ell + 1)}_i &= \sum_{j=1}^{K_1} \sum_{k=1}^{s} {\mathbf{V}}^{(\ell + 1)}_{j, k} x^{(K_1-j)m+(K_1m+X)(k-1)}_i \nonumber\\
 &+ \sum^{X}_{i=1} \mathbf{H}_i x^{(s-1)(K_1m+X)+K_1m+i-1}_i,  
\end{align}
and distributes the encoded matrix $\bar{\mathbf{V}}^{(\ell + 1)}_i$ to worker $i$ ($i \in [N-1]$). Each worker $i$ calculates $\bar{\mathbf{A}}^T_i \bar{\mathbf{V}}^{(\ell + 1)}_i$ and returns it to the master.

\textbf{Checking and Decoding Phase II}
Whenever the master receives the computation result $\mathbf{\mu}^{(\ell +1)}_i$ from the worker $i$, it will do the following checking 
\begin{equation}
\mathbf{r}^{2}_i \mathbf{\mu}^{(\ell +1)}_i = \mathbf{s}^{2}_i \bar{\mathbf{V}}^{(\ell +1)}_i.   
\end{equation}
If the equation holds, this worker passes the byzantine check. When the master receives enough results from non-byzantine workers, it will decode $\mathbf{A}^T \mathbf{V}^{(\ell +1)}$. Then the master will then transform the term $\mathbf{A}^T \mathbf{V}^{(\ell +1)}$ into $\mathbf{U}^{(\ell +1)}$. Then $\ell = \ell +1$ and we shall return to The Computation Phase I.

\begin{theorem}
The proposed scheme is memory efficient and straggler resilient. It can also maintain security and detect byzantine workers. 
\end{theorem}

\begin{proof}
We shall first explain why this is memory efficient by showing that we can reuse the encoded matrix $\bar{\mathbf{A}}_i$ in the iterative process. 
The basic point here is noticing the relationship between the encoded matrix of SEP Code and the proposed Code in section \ref{transpose sep code}. Specifically, when using SEP code to distribute the computation of $\mathbf{A} \mathbf{U}^{(\ell)}$. Recall the encoding polynomials for $\mathbf{A}$ and $\mathbf{U}^{(\ell)}$ are 
\begin{equation}
p_{\mathbf{A}}(x)=\sum_{j=1}^{K_1} \sum_{k=1}^{m} \mathbf{A}_{j, k} x^{(k-1)+(j-1)m} + \sum^{X}_{i=1} \mathbf{R}_i x^{K_1m+i-1},   
\end{equation}
and 
\begin{align}
p_{\mathbf{U}^{(\ell)}}(x) &= \sum_{j=1}^{m} \sum_{k=1}^{s} \mathbf{U}^{(\ell)}_{j, k} x^{(m-j)+(k-1)(K_1m+X)} \nonumber \\ &+ \sum^{X}_{i=1} \mathbf{T}_i x^{(s-1)(K_1m+X)+K_1m+i-1}.    
\end{align}
Thus $\bar{\mathbf{A}}_i = p_{\mathbf{A}}(x_i)$ and $\bar{\mathbf{U}}^{(\ell)}_i = p_{\mathbf{U}^{(\ell)}}(x_i)$, then from the recovery threshold of SEP code, any $(K_2+1)(K_1m+X)-1$ non-byzantine nodes' computational results are sufficient to recover the matrix multiplication $\mathbf{A}\mathbf{U}^{(\ell)}$. For the purpose of calculating the matrix product $\mathbf{A}^T\mathbf{V}^{(\ell +1)}$, we consider applying the construction in section \ref{transpose sep code} to matrices $\mathbf{A}^T$ and $\bar{\mathbf{V}}^{(\ell+1)}$. Note $\mathbf{R}_i$ has each of its components sampled from i.i.d. uniform distribution on $F_q$, thus $\mathbf{R}^T$ also has each of the components sampled from i.i.d. uniform distribution on $F_q$. Denote 
\begin{equation}
\bar{p}_{\mathbf{A}^T}(x)=\sum_{j=1}^{m} \sum_{k=1}^{K_1} (\mathbf{A}^T)_{j, k} x^{(k-1)m+j-1} + \sum^{X}_{i=1} \mathbf{R}^T_i x^{K_1m+i-1},    
\end{equation}
and 
\begin{align}
\bar{p}_{\mathbf{V}^{(\ell+1)}} &=\sum_{j=1}^{K_1} \sum_{k=1}^{s} {\mathbf{V}}^{(\ell + 1)}_{j, k} x^{(K_1-j)m +(K_1m+X)(k-1)} \nonumber \\ & + \sum^{X}_{i=1} \mathbf{H}_i x^{(s-1)(K_1m+X)+K_1m+i-1},  
\end{align}
where $\mathbf{H}_i$ is sampled from uniform distribution on $F_q$. And it is clear that $\bar{\mathbf{V}}^{(\ell+1)}_i = \bar{p}_{\mathbf{V}^{(\ell+1)}}(x_i)$. Most importantly, 
\begin{align*}
\bar{\mathbf{A}}^T_i &= (p_{\mathbf{A}}(x_i))^T \\ &= (\sum_{j=1}^{m} \sum_{k=1}^{K_1} \mathbf{A}_{k, j} x^{(k-1)m+j-1}_i + \sum^{X}_{i=1} \mathbf{R}_i x^{K_1m+i-1}_i)^T \\&=\sum_{j=1}^{m} \sum_{k=1}^{K_1} (\mathbf{A}^T)_{j, k} x^{(k-1)m+j-1}_i + \sum^{X}_{i=1} \mathbf{R}^T_i x^{K_1m+i-1}_i \\&=\bar{p}_{\mathbf{A}^T}(x_i).  \end{align*}
 This makes it exactly the form of the proposed code in section \ref{transpose sep code}, which means that we can reuse the encoded matrix $\bar{\mathbf{A}}_i$ to save the storage.

For straggler resilient and security, it follows directly from the argument of SEP and the proposed code in section \ref{transpose sep code}.

The ability to detect byzantine workers follows Freivalds’ algorithm. Note we sacrifice one worker per byzantine worker instead of two stated by the Singleton bound \cite{yu2019lagrange}.
\end{proof}

\begin{theorem} \label{complexity}
(Analysis of the Computation Complexity)
\begin{enumerate}
    \item Encoding of $\mathbf{A}$ has complexity $O(\frac{K_1m+X}{K_1m}Nab)$, this may seen as pre-processing and only needs to be done once.
    \item Generating the checking vectors $\mathbf{s}^1_i$ and $\mathbf{s}^2_i$ ($i \in [N-1]$) need a total of complexity $O(\frac{abN}{K_1m})$, this may also be seen as pre-processing and only need to be done once.
    \item The computation process of the two phases at each node will take a complexity of $O(\frac{abc}{K_1ms})$.
    \item The byzantine check process of the two phases at each node will take a complexity of $O(\frac{c}{s}(\frac{a}{K_1}+\frac{b}{m}))$. 
    \item The encoding of $\mathbf{U}^{(\ell)}$ ($\mathbf{V}^{(\ell)}$) at each node $i$ will take a complexity of $O(\frac{ms+X}{ms}Nbc)$ ($O(\frac{K_1s+X}{K_1s}Nac)$).
    \item The (worst) decoding process at each node of phase I takes a complexity of $O(\frac{ac}{K_1 s}((s+1)(K_1m+X)- 1)\log^2((s+1)(K_1m+X)- 1)\log\log((s+1)(K_1m+X)- 1)$. The (worst) decoding process at each node of phase II takes a complexity of $O(\frac{bc}{ms}((s+1)(K_1m+X)- 1)\log^2((s+1)(K_1m+X)- 1)\log\log((s+1)(K_1m+X)- 1)$.
    \item The storage overhead at each worker of storing the encoded matrix is $\frac{ab}{K_1m}$.
\end{enumerate}
\end{theorem}

\begin{proof}

\begin{enumerate}
\item The encoded matrix for $\mathbf{A}$ is given by equation (\ref{A encode}). As each sub-matrices $\mathbf{A}_{i, j}$ and $\mathbf{R}_i$ has size $\frac{ab}{K_1m}$ and the total number of sub-matrices involved in the calculation of encoded matrix at each node is $K_1m+X$. Thus the total complexity for encoding $\mathbf{A}$ is given by $O(\frac{K_1m+X}{K_1m}Nab)$.
\item The calculations of checking vectors are given by equations (\ref{check 1}) and (\ref{check 2}). Note the matrix $\bar{\mathbf{A}}_i$ has size $\frac{ab}{K_1m}$. Thus the total encoding complexity is given by $O(\frac{Nab}{K_1m})$.
\item In computation phase I, the computation at each node involves multiplying $\bar{\mathbf{A}}_i\bar{\mathbf{U}}^{(\ell)}_i$. Note $\bar{\mathbf{A}}_i$ has size $\frac{ab}{K_1m}$ and $\bar{\mathbf{U}}^{(\ell)}_i$ has size $\frac{bc}{ms}$. Thus the matrix computation complexity will be $O(\frac{abc}{K_1ms})$. The case in computation phase II is similar.
\item In computation phase I, the verification process is given by equation (\ref{verify 1}). As the matrix $\lambda^{(\ell)}_i$ has size $\frac{ac}{K_1s}$, thus the verification process has complexity of $O(\frac{ac}{K_1s})$. Computation phase II is similar.
\item The encoded matrix for $\mathbf{U}$ is given by equation (\ref{U encode}). As each sub-matrices $\mathbf{U}_{i, j}$ and $\mathbf{T}_i$ has size $\frac{bc}{ms}$ and the total number of sub-matrices involved in the calculation of encoded matrix at each node is $ms+X$. Thus the total complexity for encoding $\mathbf{U}$ is given by $O(\frac{ms+X}{ms}Nbc)$. The case for $\mathbf{V}^{(\ell +1)}$ is similar.
\item In computation phase I, the computed matrix product has size $\frac{ac}{K_1s}$. Note interpolating a degree $r$ polynomial requires a computation overhead of $O(r \log^2 r \log \log r)$ \cite{kedlaya2011fast}. As the matrix product at node $i$ can be seen as the evaluation of a degree $(s+1)(K_1m+X)- 1)$ polynomial. The conclusion follows. The case of computation phase II is similar.

\item As the encoded matrix $\bar{\mathbf{A}}_i$ is reused during the whole process. The storage overhead will be $\frac{ab}{K_1m}$.
\end{enumerate}
    
\end{proof}

\section{Applications} \label{GLM}
In this section, we shall give a brief introduction to the backgrounds and algorithms of Generalized Linear Models (GLM). From the structure of GLM, we can apply our scheme easily.

\subsubsection{Problem Formulation}
Given a dataset consisting of $m$ labelled data points $(\mathbf{x}_i, y_i) \in \mathbf{R}^n \times \mathbf{R}, i \in [n]$, we want to learn a model/parameter vector $\mathbf{w} \in \mathbf{R}^n$, which is a minimizer of the following empirical risk minimization problem:
$$
\min _{\mathbf{w} \in \mathbf{R}^n}\left(\left(\frac{1}{m} \sum_{i=1}^m f_i(\mathbf{w})\right)+h(\mathbf{w})\right),
$$
where $f_i(\mathbf{w})=\ell\left(\left\langle\mathbf{x}_i, \mathbf{w}\right\rangle ; y_i\right)$ for some differentiable loss function $\ell$. Here, each $f_i: \mathbf{R}^n \rightarrow \mathbf{R}$ is differentiable, $h: \mathbf{R}^n \rightarrow \mathbf{R}$ is convex but not necessarily differentiable, and $\left\langle\mathbf{x}_i, \mathbf{w}\right\rangle$ is the dot product of $\mathbf{x}_i$ and $\mathbf{w}$. We also denote $f(\mathbf{w}) = \frac{1}{m} \sum^{m}_{i=1} f_{i}(\mathbf{w})$.

This model includes many important algorithms as its subclass \cite{data2020data}, including Lasso, Linear programming via interior point method and Newton conjugate gradient, Logistic regression, Linear Regression, Ridge Regression, and so on. These algorithms have found many applications in scientific computing \cite{lindsey2000applying}.

\subsubsection{Main Algorithm \cite{data2020data}}

In order to update the model using the projected gradient descent algorithm, a crucial step involves the computation of the gradient vector $\nabla f(\mathbf{w})$. To facilitate this calculation, we introduce the notation 
\begin{equation}
\ell^{\prime}\left(\left\langle\mathbf{x}_{i}, \mathbf{w}\right\rangle ; y_{i}\right):=\left.\frac{\partial}{\partial u} \ell\left(u ; y_{i}\right)\right|_{u=\left\langle\mathbf{x}_{i}, \mathbf{w}\right\rangle},    
\end{equation}where $\ell^{\prime}$ represents the derivative of the loss function $\ell$ with respect to its argument, evaluated at the inner product $\left\langle\mathbf{x}_{i}, \mathbf{w}\right\rangle$. Furthermore, we denote the gradient vector $\nabla f_{i}(\mathbf{w}) \in \mathbf{R}^{d}$ as a column vector.

Let $f^{\prime}(\mathbf{w})$ be an $n$-length column vector, where the $i$-th entry corresponds to $\ell^{\prime}\left(\left\langle\mathbf{x}_{i}, \mathbf{w}\right\rangle ; y_{i}\right)$. With the aforementioned notation, a two-round approach naturally arises for computing the gradient vector $\nabla f(\mathbf{w})$ as follows:

\begin{enumerate}
\item In the first round, we compute $f^{\prime}(\mathbf{w})$ by performing a matrix-vector multiplication of $\mathbf{X}$ with $\mathbf{w}$, where $\mathbf{X}$ represents the data matrix. Each local node then computes the corresponding entry of $f^{\prime}(\mathbf{w})$ using $\mathbf{X} \mathbf{w}$, where $\mathbf{X} \mathbf{w}$ is an $n$-dimensional vector with the $i$-th entry equal to $\left\langle\mathbf{x}_{i}, \mathbf{w}\right\rangle$. Specifically, the $i$-th entry of $f^{\prime}(\mathbf{w})$ is given by $\left(f^{\prime}(\mathbf{w})\right)_{i}=\ell^{\prime}\left(\left\langle\mathbf{x}_{i}, \mathbf{w}\right\rangle ; y_{i}\right)$.

\item In the second round, we compute $\nabla f(\mathbf{w})$ by performing a matrix-vector multiplication of $\mathbf{X}^{T}$ with $f^{\prime}(\mathbf{w})$. This step involves multiplying the transpose of the data matrix $\mathbf{X}^{T}$ with the column vector $f^{\prime}(\mathbf{w})$.
\end{enumerate}

By following this two-round approach, we can accurately compute the gradient vector $\nabla f(\mathbf{w})$ necessary for the updated model using the projected gradient descent algorithm. It is very clear that the iterative process involves two costly matrix multiplications $\mathbf{X} \mathbf{w}$ and $\mathbf{X}^{T} f^{\prime}(\mathbf{w})$. As we can quantize the real numbers into the finite fields, we can apply the scheme in section \ref{structure mat mult} easily.
 
\subsubsection{Comparison}
Note the prior work \cite{data2020data} cannot guarantee security, but our scheme can. We shall then compare the scheme in \cite{data2020data} with our scheme under the same byzantine and straggler detection ability, thus $X=0$. Also, the scheme proposed in \cite{data2020data} cannot support column partition of $\mathbf{A}$, but our scheme can. Therefore, we set the column partition parameter $K_2=1$. As the there will only involve matrix-vector multiplication, thus $c=s=1$.

Suppose there will be $k$ stragglers and $t$ byzantine workers. In order to maintain straggler resilience and byzantine robustness throughout the two computation rounds during each iteration. Thus imposing a constraint of $N-k-t \geq 2K_1-1$. Therefore we take $K_1 = \frac{N+1-(k+t)}{2}$. We can obtain the computation complexity of our scheme when applied to GLM as a special case of the theorem \ref{complexity}. Combined with the complexity analysis given by \cite{data2020data}, we can summarize the comparison result in table \ref{comparison 3}.

\begin{table}[!htbp]
\caption{Comparison of different schemes to compute GLM distributedly. Define $\theta=\frac{N+1-(t+k)}{2}$.}
    \centering
    \begin{tabular}{|l|l|l|}
    \hline
         & Method in \cite{data2020data} & Our proposed method \\ \hline
        Storage overhead & $\frac{2ab}{N-2(t+k)}$ & $\frac{2ab}{N+1-(t+k)}$ \\ \hline
        Phase I check cost& $O(\frac{2a(t+k)N}{N-2(t+k)})$ & $O(a)$ \\ \hline
        Phase II check cost &  $O(\frac{2b(t+k)N}{N-2(t+k)})$ & $O(b(N-(t+k)))$\\ \hline
        Phase I decoding cost & $O(\frac{aN^2}{N-2(t+k)})$ & $O(a \log^2 \theta \log \log \theta)$ \\\hline
        Phase II decoding cost & $O(\frac{bN^2}{N-2(t+k)})$ & $O(b\theta \log^2(\theta) \log \log (\theta))$ \\ \hline
    \end{tabular}
\label{comparison 3}
\end{table}

We shall discuss the comparison in detail and organize it as follows. 
\begin{enumerate}
    \item Prior work \cite{data2020data} needs a total of $\frac{2ab}{N-2(t+k)}$ space to store the encoded matrix. Our scheme needs a total of $\frac{2ab}{N+1-(t+k)}$ space, which is strictly smaller.
    \item In phase I, our scheme check byzantine worker one by one, but prior work must wait until all the required data. We have a total of $O(a)$ complexity to detect the adversaries, while \cite{data2020data} needs $O(\frac{a}{N-2(t+k)}2(t+k)N)$. Thus our complexity is strictly smaller. Phase II is similar, except that \cite{data2020data} needs $O(\frac{b}{N-2(t+k)}2(t+k)N)$ complexity and ours demand $O(b(N-(t+k)))$.
    \item To simplify the expression, we define $\theta=\frac{N+1-(t+k)}{2}$. In phase I, prior work \cite{data2020data} needs a complexity of $O(\frac{aN^2}{N-2(t+k)})$, while we need only $O(a \log^2 \theta \log \log \theta)$, which is strictly smaller. In phase II, prior work \cite{data2020data} needs a complexity of $O(\frac{bN^2}{N-2(t+k)})$, while we need only $O(b\theta \log^2(\theta) \log \log (\theta))$. The decoding complexity of our proposed method is asymptotic similar to that of \cite{data2020data} in phase II.
\end{enumerate}

\section{Conclusion} \label{sec:conclu}
In this paper, we discussed the problem of coded matrix multiplication problem. We proposed a modulo technique that could be applied to various coded matrix designs and reduced their recovery threshold. We also discussed how to alleviate this technique in choosing the interpolation points for coded matrix multiplication. We further apply our proposed method to tackle a matrix computing problem with a special structure, which includes GLM as a special case. It remains interesting to discover locally recoverable construction that is secure and has a better recovery threshold.

\bibliographystyle{IEEEtran}
\bibliography{IEEEabrv,Bibliography}


\end{document}